\documentclass{llncs}

\usepackage[T1]{fontenc}
\usepackage[utf8]{inputenc}
\usepackage[dvips]{graphicx}
\usepackage{moreverb}
\usepackage{psfrag}
\usepackage{amsmath}
\usepackage{amsfonts}
\usepackage{amssymb}
\usepackage{hyperref}
\usepackage{authblk}

\usepackage{color}
\usepackage{framed}

\ifx\suppress\undefined
\newcommand{\TODO}[1]{
\typeout{WARNING!!! there is still a TODO left}
\marginpar{\textbf{!TODO: }\emph{#1}}
}
\else
\newcommand{\TODO}[1]{}
\fi

\ifx\suppress\undefined
\newenvironment{todo}[1]{\noindent\rule{.3\textwidth}{1pt}\TODO{#1}\\}{\\\rule{.3\textwidth}{1pt}}
\else

\fi

\ifx\suppress\undefined
\newcommand{\NOTE}[1]{
\typeout{WARNING!!! there are still DRAFT NOTES left}
\marginpar{!DRAFT}\emph{\textbf{DRAFT NOTES:} #1}
}
\else
\newcommand{\NOTE}[1]{}
\fi

\newtheorem{defi}{Definition}
\newtheorem{theo}{Theorem}
\newtheorem{lemm}{Lemma}
\newtheorem{coro}{Corollary}

\title{Parameterized~Complexity~of Weak~Odd~Domination~Problems}
\author{David Cattan\'eo\inst{2}\and Simon Perdrix \inst{1}\inst{2}}
\institute{CNRS \and LIG, University of Grenoble, France}
\date{}

\begin{document}

\maketitle

\begin{abstract}
Given a graph $G=(V,E)$, a subset $B\subseteq V$ of vertices is a \emph{weak odd dominated} (WOD) set if there  exists $D \subseteq V {\setminus} B$ such that every vertex in $B$ has an odd number of neighbours in $D$. $\kappa(G)$ denotes the size of the largest WOD set, and $\kappa'(G)$ the size of the smallest non-WOD set. 
The maximum of $\kappa(G)$ and $|V|-\kappa'(G)$, denoted $\kappa_Q(G)$, plays a crucial role in quantum cryptography. In particular deciding, given a  graph $G$ and $k>0$, whether $\kappa_Q(G)\le k$ is of practical interest in the design of graph-based quantum secret sharing schemes. The decision problems associated with the quantities $\kappa$, $\kappa'$ and $\kappa_Q$ are known to be NP-Complete.  
In this paper, we consider the approximation of these quantities and the parameterized complexity of the corresponding problems. We mainly prove the fixed-parameter intractability (\textup W$[1]$-hardness) of these problems. Regarding the approximation, we show that $\kappa_Q$, $\kappa$ and $\kappa'$ admit a constant factor approximation algorithm, and that $\kappa$ and $\kappa'$ have no polynomial approximation scheme unless \textup{P=NP}. 
\end{abstract}

\section{Introduction}

The odd domination problem is a domination type problem on a graph $G=(V,E)$ that consists of finding a set $D$ of vertices such that each vertex has an odd number of neighbours in $D$, equivalently each vertex is in the close odd neighbourhood of $D$ defined as $Odd[D]=\{u \in V, |N[u]\cap D| = 1 \bmod 2\}$. 
The odd domination falls into the general framework of $(\sigma,\rho)$-domination \cite{mod2domination,complexDomination}. The parameterized complexity of these problems has been studied, in particular in the parity cases \cite{sigmarho,W2}. 
Weak odd domination is a variation of odd domination, which does not fall into the general framework of $(\sigma,\rho)$-domination. Given a graph $G=(V,E)$, a \emph{Weak Odd Dominated} (WOD) set is a set $B\subseteq V$ such that there exists $D \subseteq V\setminus B$ with $B\subseteq Odd(D):=\{v\in V\setminus D, |N(v)\cap D|=1\bmod 2\}$, in other words, every vertex in $B$ has an odd number of neighbours in $D$. The Lemma 1 in \cite{weakodddomination} gives a good characterization of non-WOD sets: $B\subseteq V$ is not WOD if and only if $\exists C \subseteq B$ such that $|C|= 1 \bmod 2$ and $Odd(C) \subseteq B$. Since a subset of a WOD set is WOD and a superset of a non-WOD is non-WOD, we focus on the largest WOD set and the smallest non-WOD set:

\begin{defi}
Given a graph $G=(V,E)$,\\ \centerline{~~$\kappa(G)= \displaystyle{\underset{B\,  \textup{WOD}}{\max}}|B|=\displaystyle{\underset{D\subseteq V}{\max}}|Odd(D)|$\hfill $\kappa'(G)=\displaystyle{\underset{B\, \neg\textup{WOD}}{\hspace{-0.5cm}\min}}\hspace{-0.5cm}|B|=\displaystyle{\underset{C\subseteq V,|C|= 1 \bmod 2}{\hspace{-1.5cm}\min}}\hspace{-1.5cm}|C{\cup} Odd(C)|$~~}
\end{defi}

Weak odd dominated sets have a simple interpretation in a variant of the sigma-game with \emph{fragile} bulbs: given a graph $G$, to each vertex is attached a bulb which has three possible states: `on', `off', and `broken'; when one plays on a  bulb, it makes this bulb `broken' and flips the states `on'/`off' of its neighbours. In the initial configuration all bulbs are `off'. 
The size $\kappa(G)$ of the largest WOD set corresponds to the largest number of (unbroken) `on' bulbs one can obtain. Indeed, when one plays a set $D$ of bulbs, $Odd(D)$ is the set of `on' bulbs.\\
\indent Weak odd domination is strongly related to graph-based quantum secret sharing protocols, defined in \cite{quantSecret}. These protocols are represented by graphs in which every vertex represents a player. It has been proved in \cite{weakodddomination}, that for a quantum secret sharing protocol based on a graph $G$ of order $n$, $\kappa_Q(G)$ defined as $\max(\kappa(G),n{-}\kappa'(G))$ is the minimal threshold such that any set of more than $\kappa_Q(G)$ players can recover the secret. Graphs with a small quantum threshold (i.e.~$\kappa_Q(G){\le} 0.811n$ for a graph $G$ of order $n$) have been proved to exist using non constructive methods \cite{weakodddomination}. In fact, a random graph has a small $\kappa_Q$  with high probability (see  \cite{weakodddomination} for details). Thus, deciding whether a graph has a small threshold is crucial  for the generation of good graph-based quantum secret sharing protocols. Unfortunately this problem has been proved to be NP-complete~\cite{weakodddomination}. \\
\indent Since the decision problem associated with $\kappa_Q$, as well as those associated with $\kappa$ and $\kappa'$ are NP-complete \cite{weakodddomination}, we consider two approaches to tackle the hardness of these problems:  parameterized algorithms and approximations.

\noindent {\bf Parameterized Complexity.} Several NP-Complete problems, like deciding whether a graph of order $n$ has a vertex cover of size at most $k$, have been proved to be fixed parameter tractable, i.e. they can be solved in time $f(k).n^{O(1)}$ for some computable function $f$. The parameterized complexity hierarchy \cite{monograph}:  $\textup{FPT} \subseteq W[1] \subseteq W[2] \subseteq \ldots \subseteq XP$ is, roughly speaking, a way to distinguish the problems which are fixed parameter tractable (\textup{FPT}) from those which are not.

\noindent {\bf Approximation.} An optimization problem belongs to \textup{APX} if it admits a constant factor approximation  algorithm. It admits a polynomial approximation scheme if for any $\epsilon>0$ it admits a $(1+\epsilon)$-approximation \cite{approxbook}. \\
\indent In this paper, the approximation and parameterized complexity of weak odd domination problems are explored.
Section \ref{sec:bounds} is dedicated to bounds on the weak odd domination. We prove a lower bound on $\kappa$ and an upper bound on $\kappa'$ using probabilistic methods. We also prove a strong duality property between $\kappa$ and $\kappa'$. 
In section \ref{sec:PC} we define the parameterization of the weak odd domination problems, then we prove using cyclic reductions that \textsc{WOD Set Of Size At Least} $n{-}k$, \textsc{Non-WOD Set Of Size At Most} $k$ and \textsc{Quantum Threshold At Least} $n{-}k$ are equivalent to \textsc{Oddset}, which has been proved to be hard for \textup W$[1]$ and to belong to \textup W$[2]$ in \cite{oddset}, even for bipartite graphs, thus they have no \textup{FPT}-algorithms unless the parameterized hierarchy collapses with \textup W$[1]=\textup{FPT}$. 
Finally, in section \ref{sec:approx} we define the optimisation version of the weak odd domination problems and then prove that \textsc{Max non-Accessible Set} is in $\textup{APX}$ and that \textsc{Max WOD Set} and \textsc{Min non-WOD Set} are complete for $\textup{APX}$, by reduction from \textsc{Max 3-SAT B} which is hard for $\textup{APX}$ \cite{papa}, thus \textsc{Max WOD Set} and \textsc{Min non-WOD Set} have no polynomial approximation scheme unless \textup{P=NP}.

\section{Bounds on WOD sets}\label{sec:bounds}

In this section we improve the known bounds on the  largest WOD and smallest non-WOD sets in a graph. These improved bounds are essential for the choice of the parameterization of the corresponding problems. 
The largest WOD set of a graph $G$ of order $n$ and degree $\Delta$ satisfies $\Delta{\le} \kappa(G){\le}  \frac{n\Delta}{\Delta {+} 1}$ \cite{weakodddomination}. The bound $\kappa(G){\ge} \Delta$ is coming from the simple fact that any vertex is oddly dominating its neighbourhood. We improve this bound 
using  probabilistic methods.

\begin{lemm}\label{lem:bound}
For any graph $G$ of order $n$ and minimal degree $\delta>0$,\\ \centerline{$\kappa(G)\ge \left(\frac12 - \frac{1+\log(2\delta)}{4\delta}\right)n$}
Moreover, $\kappa(G)\ge \frac n 4$ when $\delta\ge 1$, and  $\kappa(G) \ge  \frac {8n}{27}$ when $\delta \ge 2$.
\end{lemm}

\begin{proof}The proof consists in evaluating the expected size of the odd neighbourhood of a randomly chosen set of vertices. 
Given $q\in [0.5,1]$, let $D$ be a subset of vertices created by choosing each $v\in V(G)$ independently with probability $1{-}q$. The expected size of $D$ is $(1{-}q)n$. For every $v$, the probability that $v\in Even(D):=\{u\in V\setminus D, |N(u)\cap D] = 0\bmod 2\}$, is $P_0(v) =q.\sum_{k=0}^{\delta(v)/2}{\delta(v)\choose 2k} (1{-}q)^{2k} q^{\delta(v){-}2k}$, and  the probability that $v\in Odd(D)$ is $P_1(v) = q.\sum_{k=0}^{\delta(v)/2}{\delta(v)\choose 2k{+}1} (1{-}q)^{2k{+}1} q^{\delta(v){-}2k{-}1}$. Notice that $v$ is not in $D$ with probability $P_0(v){+}P_1(v) = q$, moreover $P_0(v) {-} P_1(v) = q \sum_{k=0}^{\delta(v)} {\delta(v)\choose k} (q{-}1)^{k} q^{\delta(v){-}k} = q(2q{-}1)^{\delta(v)}$. As a consequence, $P_1(v) = \frac12(q{-}q(2q{-}1)^{\delta(v)})$, and the expected size of $Odd(D)$ is $E[|Odd(D)|]=\sum_{v\in V(G)}\frac 1 2(q{-}q(2q{-}1)^{\delta(v)})$. Let $x = 2q{-}1$ and $\delta$ the minimal degree of $G$, $E[|Odd(D)|] \ge \frac {n} 4(x{+}1)(1{-}x^{\delta})$, which is maximal for $x=0$ when $\delta =1$, so $E[|Odd(D)|]\ge \frac n 4$. Thus there exists $D\subseteq V(G)$ such that $|Odd(D)|\ge  \frac n 4$ so $\kappa(G)\ge \frac n 4$.   When $\delta =2$,  $\frac {n} 4(x{+}1)(1{-}x^{2})$ is maximal for $x=\frac13$, so $E[|Odd(D)|]\ge\frac{8n}{27}$. In the general case, $E[|Odd(D)|]\ge \frac n 4(1{+}x{-}2x^\delta)$, which is maximal for $x = (2\delta)^{{-}\frac 1 {\delta{-}1}}$. So $E[|Odd(D)|] \ge  \frac n 4(1{+}(2\delta)^{{-}\frac 1 {\delta{-}1}}{-}2(2\delta)^{{-}\frac \delta {\delta{-}1}}) = \frac n 4(1{+}\frac{\delta {-}1}{\delta}e^{{-}\frac{\log(2\delta)}{\delta {-}1}})\ge \frac n 4 (1{+} \frac{\delta {-}1}{\delta}(1{-}\frac{\log(2\delta)}{\delta {-}1})) = n(\frac 12 {-}\frac {1{+}\log(2\delta)}{4\delta})$. \hfill $\Box$
\end{proof}

The bound of lemma \ref{lem:bound} is not known to be tight. For the graph $C_5^k$ of order $n=5k$ composed of the disjoint union of $k$ $C_5$, $\kappa(C_5^k)= 2n/5$. Regarding connected graphs,  the largest WOD set  of a comb graph of order $2k$ (a path $P_k$ with a pending vertex on every vertex of the path) is of size $k$.  We conjecture that for any connected graph $G$, $\kappa(G) \ge \lfloor n/2\rfloor$.

Most of the graphs of order $n$ have no WOD set larger than $0.811n$. Indeed, theorem 8 in \cite{weakodddomination} implies that a random graph $G(n,1/2)$ (graph of order $n$ where every possible pair of vertices has an edge with probability 1/2), $Pr(\kappa(G(n,1/2))\le 0.811n)\ge 1-\frac1n$.

Similarly to the largest WOD set, the smallest non-WOD set of a graph $G$ of order $n$ and minimal degree $\delta$ satisfies $\frac{n}{\delta {+} 1} \le \kappa'(G)\le \delta{+}1$ \cite{weakodddomination}. The bound $\kappa'(G)\le \delta + 1$ is coming from the  fact that any vertex together with its neighbourhood is not a WOD set. 
Notice that a similar probabilistic technique as of the proof of Lemma \ref{lem:bound} fails to improve this bound: the expected size of $D\cup Odd(D)$ for a randomly chosen subset $D$ does not produce an upper bound on the smallest non-WOD set because of the additional constraint that $D$ must be of odd size. Instead, we improve the upper bound for the smallest non-WOD set by strengthening the duality property $\kappa'(G){+}\kappa(\overline G){\ge} n$  proved in \cite{weakodddomination} as follows: 

\begin{lemm}\label{lem:strdual}
For any graph $G$ of order $n$, 
$n-\kappa(\overline G)\le \kappa'(G)\le n-\frac{\kappa(\overline G)}2$.
\end{lemm}

\begin{proof} The proof consists in showing that for any graph $G$, $\kappa'(\overline G) {\le} n{-}\frac{\kappa(G)}2$. To this end, first we show that $\exists D{\subseteq} V(G)$ s.t.~$|D|{=}1\bmod 2$ and $|Odd(D)|{\ge} \frac{\kappa(G)}2$. Indeed let  $D{\subseteq}V(G)$ be a non empty set s.t.~$|Odd(D)| {=} \kappa(G)$. If $|D|{=}1 \bmod 2$ then  we obviously have $|Odd(D)|{\ge} \frac{\kappa(G)}2$. Othewise, if $D$ is of even size then $\forall v{\in} D$,  $|N(v)|{+}|Odd(D{\setminus} \{v\})|{\ge} |Odd(D)|{=}\kappa(G)$. So either $\{v\}$ or $D{\setminus} \{v\}$, which are both of odd size, has an odd neighbourhood larger than $\frac{\kappa(G)}2$. 
Thus, $\exists C{\subseteq} V(G)$ s.t. $|C| {=}1\bmod2$ and $|Odd(C)|{\ge} \frac{\kappa(G)}2$. Since $|C|{=}1\bmod 2$ implies that $\forall v\notin C$, $v{\in} Odd(C) \Leftrightarrow v{\notin} Odd_{\overline G} (C)$ (where  $Odd_{\overline G} (C)$ is the odd neighbourhood of $C$ in $\overline G$), $|C{\cup} Odd_{\overline G}(C)|{\le} n{-}\frac{\kappa(G)}2$. Thus $\kappa'(\overline G){\le} |C{\cup} Odd_{\overline G}(C)|{\le} n{-}\frac{\kappa(G)}2$. $\hfill \Box$ \end{proof}

\begin{coro}\label{cor:bound}
For any graph $G$ of order $n$ and degree $\Delta<n-1$, \\ \centerline{$\kappa'(G)\le \frac {7n}8~~~~~\text{and}~~~~~\kappa'(G)\le \left(\frac34 + \frac{1+\log(2(n-\Delta-1)}{8(n-\Delta-1)}\right)n$}

\end{coro}

The restriction $\Delta<n-1$ in Corollary \ref{cor:bound} is crucial since  for any $n$, $\kappa'(K_n)=n$, where $K_n$ is the complete graph of order $n$. Similarly, the condition $\delta>0$ in lemma \ref{lem:bound} is necessary since $\kappa(\overline {K_n}) = 0$. We consider the class of graphs which do not satisfy these conditions, i.e. graphs having a universal vertex or an isolated vertex, and show that such graphs satisfy a stronger duality property:

\begin{lemm}\label{lem:stgdual}
For any graph $G$ of order $n$ with a universal or isolated vertex,\\ \centerline{$\kappa'(G)+\kappa(\overline G)=n$}
\end{lemm}

\vspace{-0,5cm}

\begin{proof}
If $G$ has an isolated vertex then $\kappa'(G){=}1$ and $\kappa(\overline G){=}n{-}1$. Otherwise, let $u$ be a universal vertex in $G$.  Let $ D{\subseteq} V(G)$ s.t. $|Odd_{\overline G} (D)|{=}\kappa(\overline G)$. Notice that $u$ is isolated in $\overline G$, so $|Odd_{\overline G}(D{\oplus} \{u\})| {=} |Odd_{\overline G}(D)|{=}\kappa(\overline G)$, where ${\oplus} $ denotes the symmetric difference. Since either $D$ or $D{\oplus} \{u\}$ is of odd size, $\exists C{\in} \{D, D{\oplus} \{u\}\}$ s.t. $|C|{=}1\bmod 2$ and $|Odd_{\overline G} (C)|{=}\kappa(\overline G)$. Moreover $|Odd_{G}(C)| {=} n{-}\kappa(\overline G)$,  so $\kappa'(G) \le n{-}\kappa(\overline G)$, which implies  $\kappa'(G) {=} n{-}\kappa(\overline G)$ since $\kappa'(G) {\ge}  n{-}\kappa(\overline G)$ for any graph. \hfill $\Box$
\end{proof}

The strong duality property gives a way to get ride of the universal and isolated vertices in the computation of $\kappa(G)$ and $\kappa(\overline G)$. For instance, if $G$ has a universal vertex $u$, $\kappa'(G)= n-\kappa(\overline G) = n-\kappa(\overline G\setminus u)$, since $u$ is isolated in $\overline G$. 

\section{Parameterized Complexity}\label{sec:PC}

\subsection{Parameterization of weak odd domination problems}

The results of the previous  section imply that when parameterised by the size of the WOD set, the largest WOD set problem is in \textup{FPT}. More precisely, given a graph $G$ and a parameter $k$, deciding  whether $\kappa(G)\ge k$ is  fixed parameter tractable using the following algorithm {\verb@Kappa@}$(G,k)$:

\vspace{0.3cm}

\noindent$\bullet$ Remove all isolated vertices. Let $n$ be the order of the resulting graph $G'$
\\$\bullet$ If $k\le \frac n4$ then `true'
\\$\bullet$ Else if  $\forall D\subseteq V(G')$,  $|Odd(D)| < k$ then `false' else `true'. 

\vspace{0.3cm}

The first step of the algorithm is correct since for any isolated vertex $u$, $\kappa(G)=\kappa(G\setminus \{u\})$. The complexity of the algorithm is $O^*(2^n)$ and since in the worst case $k=\frac{n}{4}$ the complexity in fonction of $k$ is $O^*(2^{4k})$, so the problem WOD of size at least $k$ is \textup{FPT}.

Regarding the smallest non-WOD set,  given a graph $G$ of order $n$ and a parameter $k$,  deciding whether $\kappa'(G)\le n{-}k$ is  fixed parameter tractable  using the following algorithm {\verb@Kappa'@}$(G,k)$:

\vspace{0.3cm}

\noindent$\bullet$ If $G$ has a universal vertex $u$ then {\verb@Kappa@}$(\overline G\setminus u,k)$
\\ $\bullet$ Else if $k\le n/8$ then `true'
\\   $\bullet$ ~~~~~~Else if  $\forall D$ of odd size $|D\cup Odd(D)| > n- k$ then `false' else  `true'.

\vspace{0.3cm}

The correctness of the first step is based on Lemma \ref{lem:stgdual}: if $G$ has a universal vertex $u$, $\kappa'(G)= n-\kappa(\overline G) = n- \kappa(\overline G\setminus u)$. 

\vspace{0.3cm}

The fixed parameter tractability of the problems WOD of size at least $k$ and non-WOD of size at most $n{-}k$ is not relevant because only based on the existence of  bounds for $\kappa(G)$ and $\kappa'(G)$. As a consequence, we  focus in the rest of this paper on the dual parameterization of these problems: 

\vspace{0.3cm}

\indent {\begin{tabular}{l}
 \textsc{WOD Set Of Size At Least} $n{-}k$~~\\
   input: A graph $G$ of order $n$\\
  parameter: An integer $k$\\
  question: Is $\kappa(G) \geq n{-}k$?\\
  \end{tabular}}
  
  \vspace{0.3cm}

  \indent {\begin{tabular}{l}
\textsc{non-WOD Set Of Size At Most} $k$\\
input: A graph $G$ of order $n$\\
 parameter: An integer $k$\\
 question: Is $\kappa'(G) \le k$?\\
  \end{tabular}}

\vspace{0.3cm}
Concerning the quantum threshold problem, given a graph $G$ of order $n$ and a parameter $k$ deciding whether $\kappa_Q(G){=} \max (\kappa(G), n{-} \kappa'(G)){\ge} k$ is in \textup{FPT} 
 since $\kappa_Q(G){\ge} 0.506n$ \cite{protocols}. As a consequence we consider its dual parameterization: 

\vspace{0.3cm}

  \indent {\begin{tabular}{l}
\textsc{Quantum Threshold At Least} $n{-}k$\\
 input: A graph $G$ of order $n$\\
 parameter: An integer $k$\\
 question: Is $\kappa_Q(G) \geq n{-}k$?\\
  \end{tabular}}
%
%
\subsection{Fixed parameter intractability}

In this section we show that the three problems  \textsc{WOD Set Of Size At Least} $n{-}k$, \textsc{non-WOD Set Of Size At Most} $k$, and \textsc{Quantum Threshold At Least} $n{-}k$ are all hard for \textup W$[1]$ and belong to the class \textup W$[2]$. \textup W$[1]$-hardness implies the fixed parameter intractability of these problems (unless \textup{FPT} = \textup W$[1]$). \textup W$[2]$ membership and \textup W$[1]$-hardness of these problem are proved using reductions from the \textsc{Oddset} problem:

\vspace{0.3cm}

\noindent ~~\textsc{Oddset of size at most} $k$

\noindent ~~input: A bipartite graph $G=(R \cup B,E)$

\noindent ~~parameter: An integer $k$

\noindent ~~question: Is there a subset $R'\subseteq R$, $|R'|\le k$ such that  $B = Odd(R')$?

\vspace{0.3cm}

This problem is known  \cite{oddset} to be \textup W$[1]$-hard 
and \textup W$[2]$. We prove the following circular reductions,  where $B\to A$ stands for $A$ is \textup{FPT}-reducible to $B$:\\

~~{\begin{psfrags}
\psfrag{o}[c]{\textsc{Oddset}}
\psfrag{q}[c]{\textsc{~~Largest WOD Set}}
\psfrag{S}[c]{\textsc{Smallest}}
\psfrag{N}[c]{\textsc{non-WOD Set}}
\psfrag{L}[c]{\textsc{Quantum Threshold}}
\psfrag{t1}[l]{\footnotesize [Theorem \ref{thm:WOD_ODD}]}
\psfrag{t2}[l]{\footnotesize [Theorem \ref{thm:nonWOD_ODD}]}
\psfrag{t3}[r]{\footnotesize [Theorem \ref{thm:quantumt}]}
\psfrag{t4}[r]{\footnotesize [Theorem \ref{thm:qodd}]}
\centerline{\includegraphics[width=11cm,height=4cm]{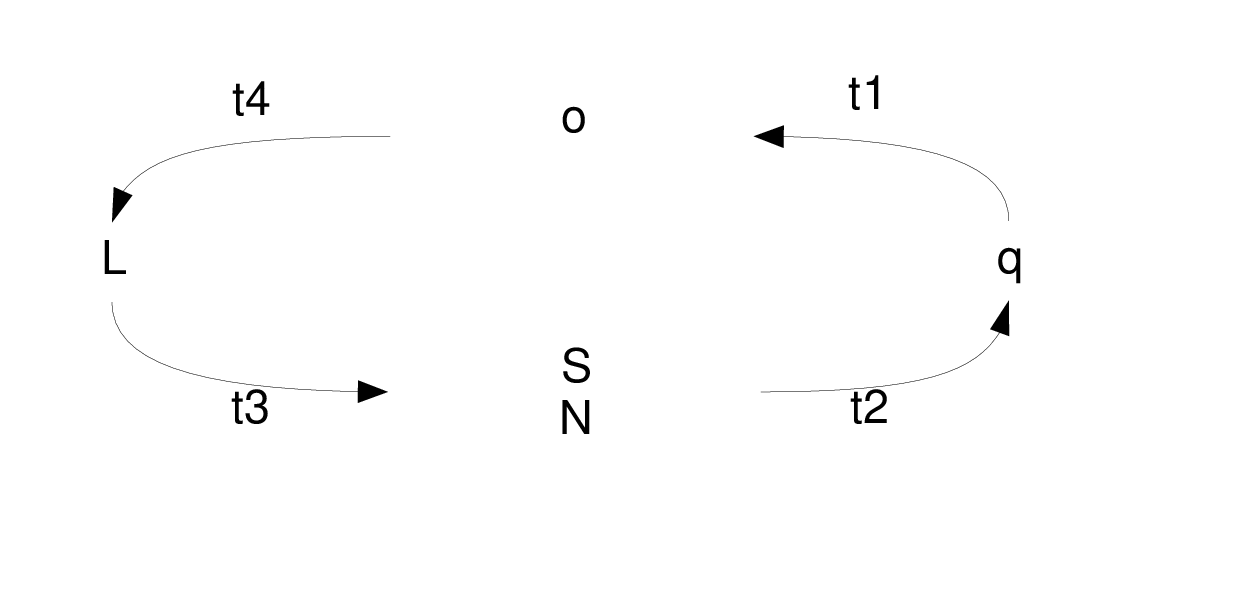}}
\end{psfrags}}

As a consequence, all these problems are \textup{FPT}-equivalent to \textsc{Oddset}. Moreover we show that they remain \textup{FPT}-equivalent to $\textsc{Oddset}$ when  restricted to bipartite graphs. 

\subsubsection{3.2.1 Largest WOD set problem}\label{sec:WOD}

\begin{theo}\label{thm:WOD_ODD}
\textsc{WOD Set Of Size At Least} $n{-}k$ is harder than \textsc{Oddset} by an \textup{FPT}-reduction.
\end{theo}

\begin{proof}
Given $(G,k)$ where $G{{=}}(R{{\cup}} B,E)$ an instance of \textsc{Oddset} let  $(G',k')$ (see Figure \ref{fig:WOD_ODD}) be  an instance of \textsc{WOD Set Of Size At Least} $n'{-}k'$ such that $G'{=}(A {\cup} D {\cup} F {\cup} c$, $E_1 {\cup} E_2 {\cup} E_3)$, $n'{=}|R|{+}(k{+}2)|B|{+}(k{+}2){+}1$ and $k'{=}k{+}1$ where:\\
$\begin{array}{ll}
A{=} \{ a_u, u {\in} R\}&E_1 {=} \{ cf_i, 1 \leq i \leq k{+}2 \}\\
D{=} \{ d_{u,i}, u {\in} B, 1 \leq i \leq k{+}2 \} &E_2 {=} \{ ca_u, u {\in} R \}\\
F {=} \{ f_i, 1 \leq i \leq k{+}2 \} &E_3 {=} \{ a_ud_{v,i}, uv {\in} E, 1 \leq i \leq k{+}2 \}\\
\end{array}$\\

If $(G',k')$ is a positive instance of \textsc{WOD Set Of Size At Least} $n'{-}k'$, let $C{{\subseteq} }V(G')$ be the smallest set such that $|Odd(C)|{\ge} n'{-}k'$. $F$ is an independent set of size $k{{+}}2{>}k'$ so there exists $f{\in} F{\cap} Odd(C)$. Since $c$ is the unique neighbour of $f$, $c$ belongs to $C$. For every $u{\in} B$, the subset $D_u{=}\{d_{u,i}, 1{\le}  i{\le}  k{+}2\}$ is an independent set of size $k{+}2>k'$ whose neighbourhood is included in $A$, thus $\forall u{\in} B, D_u{\subseteq} Odd(C{\cap} A)$ so $D{\subseteq} Odd(C{\cap} A)$. Since $c{\in} C$ and $A{\subseteq} Odd(\{c\})$, by minimality of $C$, $D{\cap} C{=}\emptyset$. Let $R'{=}\{u{\in} R, a_u{\in} C\}$, of size $|R'|{=}|C{\cap} A|{=}|C|{-}1{\le} k$. Since $\forall u{\in} B, u{\in} Odd(R')\Leftrightarrow D_u{\subseteq} Odd(C)$,  $B{\subseteq} Odd(R')$ so $(G,k)$ is a positive instance of \textsc{Oddset}.

If $(G,k)$ is a positive instance of \textsc{Oddset}, there exists $R' {\subseteq}  R$, s.t. $|R'| {\le} k$ and $B{=}Odd(R')$. Let  $A' {=} \{a_u, u{\in} R'\}$, since $\forall u{\in} B, u{\in} Odd(R')\Leftrightarrow D_u{\subseteq} Odd(A')$, $D {\subseteq} Odd(A')$ so $D {\subseteq} Odd(A'{\cup} \{c\})$. Since $A{\cup} F$ is an independent set dominated by $c$, $(F{\cup} A{\cup} D){\setminus} A'{=}V(G')\setminus (A'{\cup} \{c\}){\subseteq} Odd(A'{\cup} \{c\})$. Moreover $|Odd(A'{\cup} \{c\})\ge n'{-}(k{+}1){=}n'{-}k'$, so $(G',k')$ is a positive instance of \textsc{WOD Set Of Size At Least} $n'{-}k'$. \hfill $\Box$
\end{proof}

\begin{figure}
	\centering
	\vspace{-1cm}
	\begin{psfrags}
	{\footnotesize
	\psfrag{F}{$\!F$}
	\psfrag{c}{$c$}
	\psfrag{A}{$\!A$}
	\psfrag{G}{$\bf G$}
	\psfrag{D}{$\!D$}
	\psfrag{A}[c]{$~A$}
\psfrag{B}[c]{$~B$}
	\psfrag{vdots}{\vdots}
	\includegraphics[width=5cm,height=3cm]{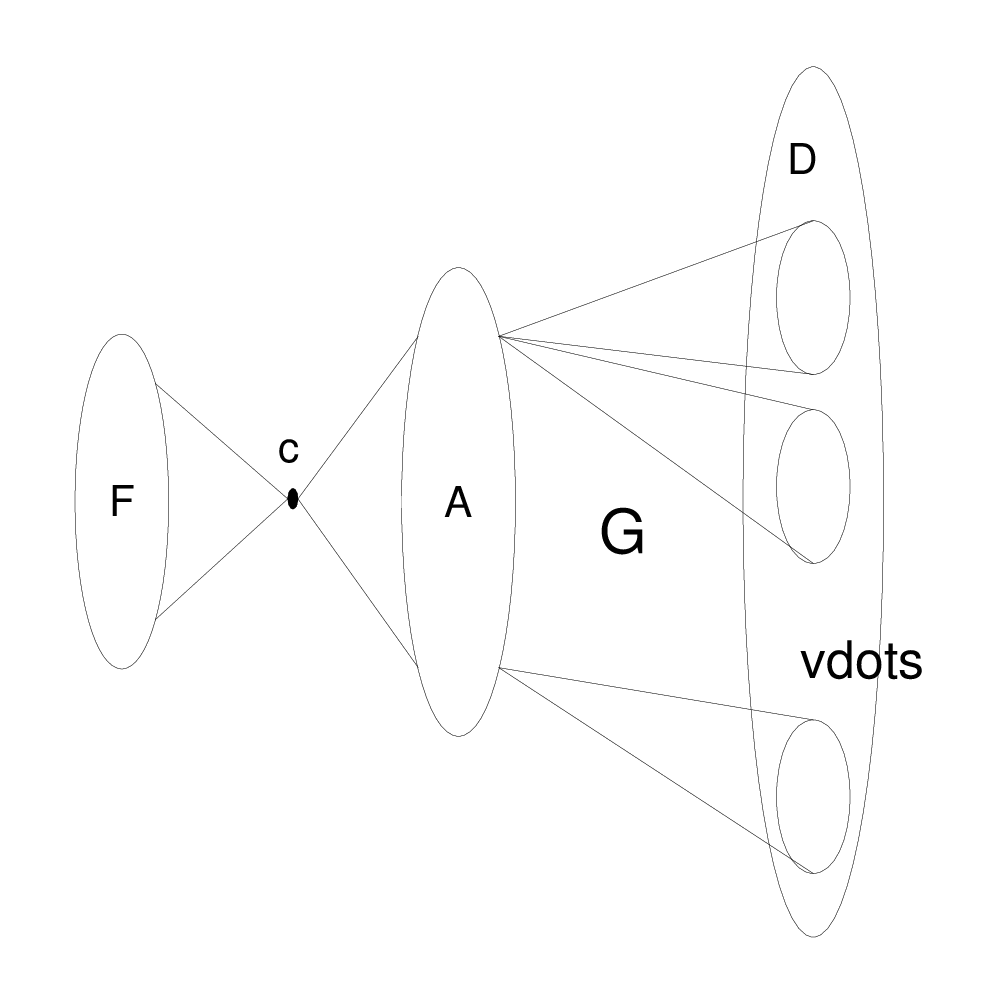}~~~~~~~~~~~~~~~~~\includegraphics[scale=0.3]{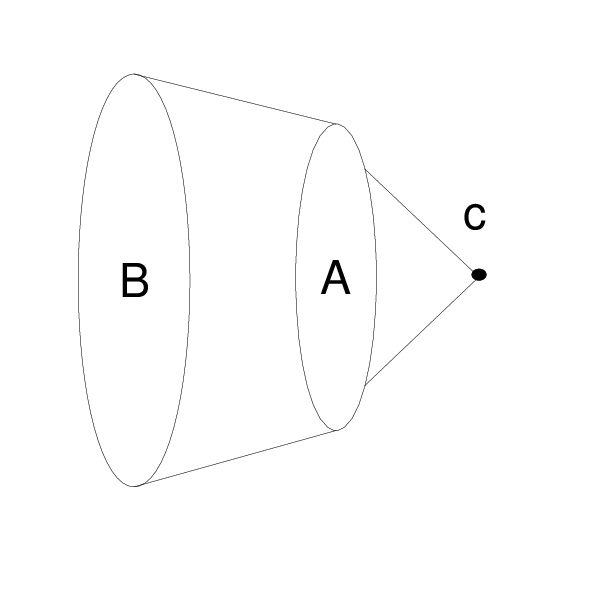}}
	\end{psfrags}
	\vspace{-0.3cm}
	\caption{	 \small\label{fig:WOD_ODD}Left: Reduction from \textsc{Oddset} to \textsc{WOD Set Of Size At Least} $n{-}k$; Right: Reduction from \textsc{WOD Set Of Size At Least} $n{-}k$ to \textsc{Non-WOD Set Of Size At Most} $k$\label{fig:kappa_kappap}\vspace{-0.3cm}} 
\end{figure}

Since \textsc{Oddset} is hard for \textup W$[1]$, so is \textsc{WOD Set Of Size At Least} $n{-}k$. Moreover, notice that the graph used in the proof of Theorem \ref{fig:WOD_ODD} is bipartite (see figure \ref{fig:WOD_ODD}), as a consequence:

\begin{coro}
\textsc{WOD Set Of Size At Least} $n{-}k$ is hard for \textup W$[1]$ even for bipartite graphs.
\end{coro}

\vspace{-0,7cm}

\subsubsection{3.2.2 Smallest non-WOD set problem}\label{sec:nWOD}
~\\
In this section we prove that \textsc{Non-WOD Set Of Size At Most} $k$ is hard for \textup W$[1]$ even for bipartite graphs.

\begin{theo}\label{thm:nonWOD_ODD}
\textsc{Non-WOD Set Of Size At Most} $k$ is harder than \textsc{WOD Set Of Size At Least} $n{-}k$ by
an \textup{FPT}-reduction.
\end{theo}

\begin{proof}
Given $(G,k)$ where $G=(V,E)$ and $n=|V|$ an instance of \textsc{WOD Set Of Size At Least} $n{-}k$, let $(G',k')$ (see Figure \ref{fig:kappa_kappap}) be an instance
of \textsc{Non-WOD Set Of Size At Most} $k'$ such that $G'=(A \cup B \cup \{c\}$, $E_1 \cup E_2 \cup E_3)$ and $k'=k+2$ where:\\
$\begin{array}{ll}
A= \{ a_i, 1 \leq i \leq k+3 \} & E_1= \{b_ub_v, uv\in E\}\\
B=\{b_u, u \in V\} & E_2= \{a_ic, 1 \leq i \leq k+3 \}\\
&E_3= \{b_ua_i, u \in V, 1 \leq i \leq k+3 \}\\
\end{array}$\\

If $(G,k)$ is a positive instance of \textsc{WOD Set Of Size At Least} $n{-}k$, there exists $C\subseteq V$ such that $|Odd(C)|\ge n{-}k$. There are two cases: if $|C| = 0 \bmod 2$, let $C'=\{b_u, u \in C\} \cup \{a\}$, where $a$ is any vertex in $A$. Since $A$ is connected by a complete bipartite graph to $B$, $A \subseteq Even(C')$ and $c \in Odd(C')$. Since $a\in C'$ is connected to all the vertices in $B$, for every $u\in Odd(C)$ in $G$, $b_u\notin Odd(C')$ in $G'$, so $|C'\cup Odd(C')|\le k+2 = k'$. Otherwise, if $|C| = 1\bmod 2$, let $C'=\{b_u,u\in C\} \cup \{a,c\}$, where $a$ is any vertex in $A$. $A$ is connected by a complete bipartite graph to $B$ and each vertex in $A$ is connected to $c$ since $|B\cap C'|= 1 \bmod 2$, $A\subseteq Even(C')$. Similarly $|Odd(C') \cup C'| \leq k+2 = k'$. Thus in both cases $(G',k')$ is a positive instance of \textsc{Non-WOD Set Of Size At Most} $k$.

\indent  If $(G',k')$ is a positive instance of \textsc{Non-WOD Set Of Size At Most} $k$, there exists $C'\subseteq V(G')$ such that $|C'\cup Odd(C')|\le k$ and $|C'|= 1 \bmod 2$. $A$ is an independent set of size $k+3>k'$, so there exists $a\in A$ such that $a\in Even(C')$. Since $A$ is connected to $V(G')\setminus A$ by a complete bipartite graph and $|C'|= 1 \bmod 2$, $|C'\cap A|= 1 \bmod 2$ then by minimality $|C'\cap A|=1$, let $a$ be this vertex. Let $C=\{u,b_u\in C'\}$, since $a$ is connected to every vertex in $B$, $\forall u\in V, b_u\in Odd(C')\Leftrightarrow u\in Even(C)$ so $|Even(C)|<k$ thus $(G,k)$ is a positive instance of \textsc{WOD Set Of Size At Least} $n{-}k$. \hfill $\Box$
\end{proof}

\begin{coro}\label{NONWOD:W1H}
\textsc{Non-WOD Set Of Size At Most} $k$ is hard for \textup W$[1]$.
\end{coro}

The proof of the \textup W$[1]$-hardness of \textsc{Non-WOD Set Of Size At Most} $k$ does not respect the bipartition of the graph. However we prove that the problem is \textup W$[1]$-hard even for bipartite graph by reduction from the general case:

\begin{theo}\label{thm:bip}
\textsc{Non-WOD Set Of Size At Most} $k$ in bipartite graphs is harder than \textsc{Non-WOD Set Of Size At Most} $k$ by an \textup{FPT}-reduction.
\end{theo}

\vspace{-0.3cm}

\begin{proof}
Given an instance $(G,k)$ of \textsc{Non-WOD Set Of Size At Most} $k$ let $(G',k')$ (see Figure) be a bipartite instance of \textsc{Non-WOD Set Of Size At Most} $k$ with:\\
$G'=( A \cup B_1 \cup B_2 \cup D \cup F \cup H,E_1 \cup E_2 \cup E_3 \cup E_4 \cup E_5)$, $k'=2k$\\
$\begin{array}{ll}
A= \{ a_u, u \in V \} & E_1= \{ a_ub_{i,v}, i \in \{1,2\}, uv \in E\}\\
B_1= \{ b_{1,u}, u \in V \} & E_2= \{ a_ub_{2,u}, u \in V\}\\
B_2= \{ b_{2,u}, u \in V \}& E_3= \{ b_{i,u}d_{i,u,j}, i \in \{1,2\}, \\
D= \{ d_{i,u,j}, i \in \{1,2\},  &\qquad\quad u \in V, 1 \leq j \leq 2k+1\}\\
\qquad\quad u \in V, 1\leq j \leq 2k+1 & E_4= \{ d_{i,u,j}f_{i,u,j,l}, i \in \{1,2\}, \}\\
F= \{ f_{i,u,j,l}, i \in \{1,2\},  &\qquad\quad u \in V, 1 \leq j,l \leq 2k+1\}\\
\qquad\quad u \in V, 1\leq j,l \leq 2k+1 \}& E_5= \{ f_{i,u,j,l}h_p, i \in \{1,2\}, \\
H= \{ h_i, 1 \leq i \leq 2k+1 \}&\qquad\quad u \in V, 1 \leq j,l,p \leq 2k+1\}\\
\end{array}$\\

	\begin{psfrags}
	\psfrag{A}{$\!A$}
	\psfrag{B1}{$\!B_1$}
	\psfrag{B2}{$\!B_2$}
	\psfrag{G}{$\bf G$}
	\psfrag{gum}{~~~$\bf G\cup M$}
	\psfrag{D}{$\!D$}
	\psfrag{F}{$\!F$}
	\psfrag{H}{$\!H$}
	\psfrag{vdots}{\vdots}
	\includegraphics[width=11cm,height=5cm]{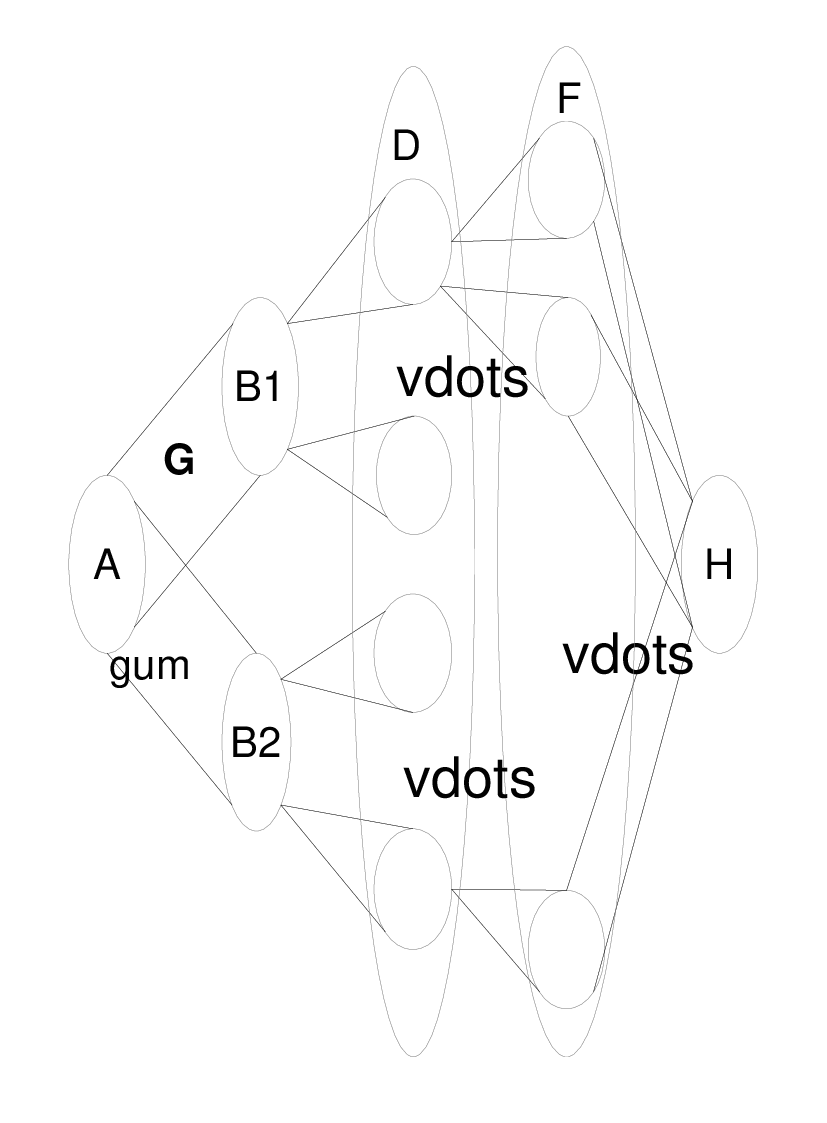}
	\end{psfrags}

\vspace{-0.2cm}
\noindent If $(G,k)$ is a positive instance of \textsc{Non-WOD Set Of Size At Most} $k$, there exists $C\subseteq V$ such that $|C \cup Odd(C)|\leq k$. Let $C'=\{a_u, u \in C\}$, notice that $|Odd(C')\cap B_1|=|Odd(C)|$ since $\forall u,v\in V,  a_ub_{1,v}\in E_1\Leftrightarrow uv\in E$. Moreover, $|Odd(C')\cap B_2|=|Odd(C){\oplus} C|$, since $\forall u,v\in V,\ a_ub_{2,v}\in E_1\Leftrightarrow uv\in E$ and $\forall u\in V, a_ub_{2,u}\in E_2$. So $|C'\cup Odd(C')|=|C|+|Odd(C)|+|Odd(C){\oplus} C|=2|C \cup Odd(C)|\leq 2k=k'$ thus $(G',k')$ is a positive  bipartite instance of \textsc{Non-WOD Set Of Size At Most} $k$.

\noindent If $(G',k')$ is a positive bipartite instance of \textsc{Non-WOD Set Of Size At Most} $k$, there exists $C'\subseteq V'$ such that $|C' \cup Odd(C')|\leq 2k$ and $|C'|= 1 \bmod 2$. Notice that $H$ is an independent set of size $2k+1>k'$, so there exists $h\in H$ such that $h\in Even(C')$, since $H$ is connected by a complete bipartite graph to $F$ thus $|F\cap C'|= 0 \bmod 2$. $F$ is composed of $(2k+1)2|V|>k'$ independent sets $F_i$ of size $2k+1>k'$ so $\forall i\in [1,(2k+1)2|V|] \exists f_i\in F_i$ such that $f_i\in Even(C')$. Since each $F_i$ is connected to a vertex $d\in D$ and connected to $H$ by a complete bipartite graph, $|F\cap C'|= 0 \bmod 2$ and $D\cap C'=\emptyset$. By minimality of $C'$ $|F\cap C'|= 0 \bmod 2$ implies $C'\cap F=\emptyset$. $D$ is composed of $2|V|>k'$ independent sets $D_i$ of size $2k+1>k'$, so there exists $i\in [1,2|V|]$ such that $D_i\subseteq Even(C')$, since each $D_i$ is connected to a vertex $b\in B_1\cup B_2$, $C'\cap (B_1\cup B_2)=\emptyset$. Thus $C'\subseteq A$ and $Odd(C')\subseteq B_1\cup B_2$, let $C=\{u\in V, a_u\in C'\}$, so $|C'\cup Odd(C')|=|C|+|Odd(C)|+|Odd(C){\oplus} C|=2|C \cup Odd(C)|$. Since $|C'\cup Odd(C')|\le k'$, $|C\cup Odd(C)|\le k$ which implies that $(G,k)$ is a positive instance of \textsc{Non-WOD Set Of Size At Most} $k$. \hfill $\Box$
\end{proof}

\begin{coro}
\textsc{Non-WOD Set Of Size At Most} $k$ is hard for \textup  W$[1]$ even for bipartite graphs.
\end{coro}


\subsubsection{3.2.3 Quantum Threshold problem}\label{sec:kappaq}
~\\
In this section we consider the quantum threshold problem. The quantum threshold $\kappa_Q(G)$ of a graph $G$ of order $n$ is defined as $\kappa_Q(G)=\max(\kappa(G),\kappa(\overline{G}))=\max(\kappa(G),n{-}\kappa'(G))$. Like the largest WOD set problem, this problem is proved to be hard for \textup W$[1]$ by a reduction from \textsc{Non-WOD Set Of Size At Most} $k$.

\begin{theo} \label{thm:quantumt}
\textsc{Quantum Threshold At Least} $n{-}k$ is harder than \textsc{Non-WOD Set Of Size At Most} $k$ by an
\textup{FPT}-reduction.
\end{theo}

\begin{proof}
Given $(G,k)$ an instance of \textsc{Non-WOD Set Of Size At Most} $k$ let $(G^{k+1},k)$ be an instance of 
\textsc{Quantum Threshold At Least} $n{-}k$ where  $G^p$ is the graph obtained by copying $p$ times $G$.
Since $\kappa_Q(G)=\max(\kappa(G),n-\kappa'(G))$ by Lemma 6 of \cite{weakodddomination}, there are two possibilities for $\kappa_Q(G^{k+1}) \geq (k+1)n{-}k$, either $\kappa(G^{k+1})$ or $n(k+1)-\kappa'(G^{k+1})$ is maximal:
\\$\bullet$ If $\kappa(G^{k+1}) \geq (k+1)n{-}k$, since the $k+1$ copies of $G$ are independent the largest WOD set of $G^{k+1}$ is $k+1$ copies of the largest WOD set of $G$, so $(k+1)\kappa(G) \geq (k+1)n - k$. Thus $\kappa(G) \geq n - \frac{k}{(k+1)}$, but $\frac{k}{k+1} < 1$ so $\kappa(G) \ge n$. On the other hand, for any graph $G$ of order $n$ we have $\kappa(G)<n$ so there is a contradiction.
\\$\bullet$ If $(k+1)n-\kappa'(G^{k+1}) \geq (k+1)n{-}k$. Since the $k+1$ copies are independent the smallest non-WOD set of $G^{k+1}$ is one of the smallest non-WOD set in a copy of $G$, so $\kappa'(G^{k+1}) \leq k\Leftrightarrow \kappa'(G) \leq k$ which is equivalent to $(G,k)$ is a positive instance of \textsc{Non-WOD Set Of Size At Most} $k$. \hfill $\Box$
\end{proof}

Notice that the reduction preserves the bipartition of the graph, so \textsc{Quantum Threshold At Least} $n{-}k$ is \textup W$[1]$-hard even for bipartite graphs.

\begin{coro}
\textsc{Quantum Threshold At Least} $n{-}k$ is hard for \textup W$[1]$ even for bipartite graphs.
\end{coro}

Now we have proved that all the problems related to weak odd domination are hard for \textup W$[1]$, by successive \textup{FPT}-reductions starting from \textsc{Oddset}. These problems are not only harder than \textsc{Oddset} but equivalent to \textsc{Oddset} and belong to \textup W$[2]$. Indeed, \textsc{Oddset} is harder than \textsc{Quantum Threshold Of Size At Least} $n{-}k$ by an \textup{FPT}-reduction:

\begin{theo}\label{thm:qodd}
\textsc{Oddset} is harder than \textsc{Quantum Threshold At Least} $n{-}k$ by an \textup{FPT}-reduction.
\end{theo}

\begin{proof}
Given an instance $(G,k)$ of \textsc{Quantum Threshold At Least} $n{-}k$, let $(G',k')$ (see Figure) be an instance of \textsc{Oddset} with:\\
$G'=(A \cup \{d_1\} \cup \{d_2\} \cup \{c\},E_1 \cup E_2 \cup E_3 \cup E_4 \cup E_5 \cup \{d_1c\} \cup \{d_2c\})$\\
$\begin{array}{l l l}
A= \bigcup \limits_{1\leq i \leq 2,1\leq j \leq 5} A_{i,j} & & E_1= \{ a_{1,2,u}a_{1,j,v}, j \in \{4,5\}, uv \in E \}\\
A_{i,j}= \{ a_{i,j,u}, u \in V \} & & E_2= \{ a_{2,2,u}a_{2,j,v}, j \in \{4,5\}, uv \notin E \}\\
& \quad & E_3= \{ a_{i,j,u}a_{i,l,u}, i {\in} \{1,2\}, j {\in} \{1,3\}, l {\in} \{4,5\}, u {\in} V \}\\
& & E_4= \{ d_ia_{i,j,u}, i \in \{1,2\}, j \in \{4,5\}, uv \in E \}\\
& & E_5= \{ a_{i,2,u}a_{i,5,u}, i \in \{1,2\}, u \in V \} \\
\end{array}$\\
And with $B=A_{i,4} \cup A_{i,5} \cup \{c\}$, $R=A_{i,1} \cup A_{i,2} \cup A_{i,3}$ and $k'=2k+1$.

	\begin{psfrags}\footnotesize
	\psfrag{A11}{$A_{1,1}$}
	\psfrag{A12}{$A_{1,2}$}
	\psfrag{A13}{$A_{1,3}$}
	\psfrag{A14}{$A_{1,4}$}
	\psfrag{A15}{$A_{1,5}$}
	\psfrag{A21}{$A_{2,1}$}
	\psfrag{A22}{$A_{2,2}$}
	\psfrag{A23}{$A_{2,3}$}
	\psfrag{A24}{$A_{2,4}$}
	\psfrag{A25}{$A_{2,5}$}
	\psfrag{d1}{$d_1$}
	\psfrag{d2}{$d_2$}
	\psfrag{c}{$c$}
	\psfrag{M}{$\bf M$}
	\psfrag{G}{$\bf G$}
	\psfrag{Gb}{$\bf \overline{G}$}
	\psfrag{Gbum}{$~~~\bf \overline{G}\cup M$}
	\psfrag{Gum}{$\bf G\cup M~$}
	\includegraphics[width=11cm,height=6cm]{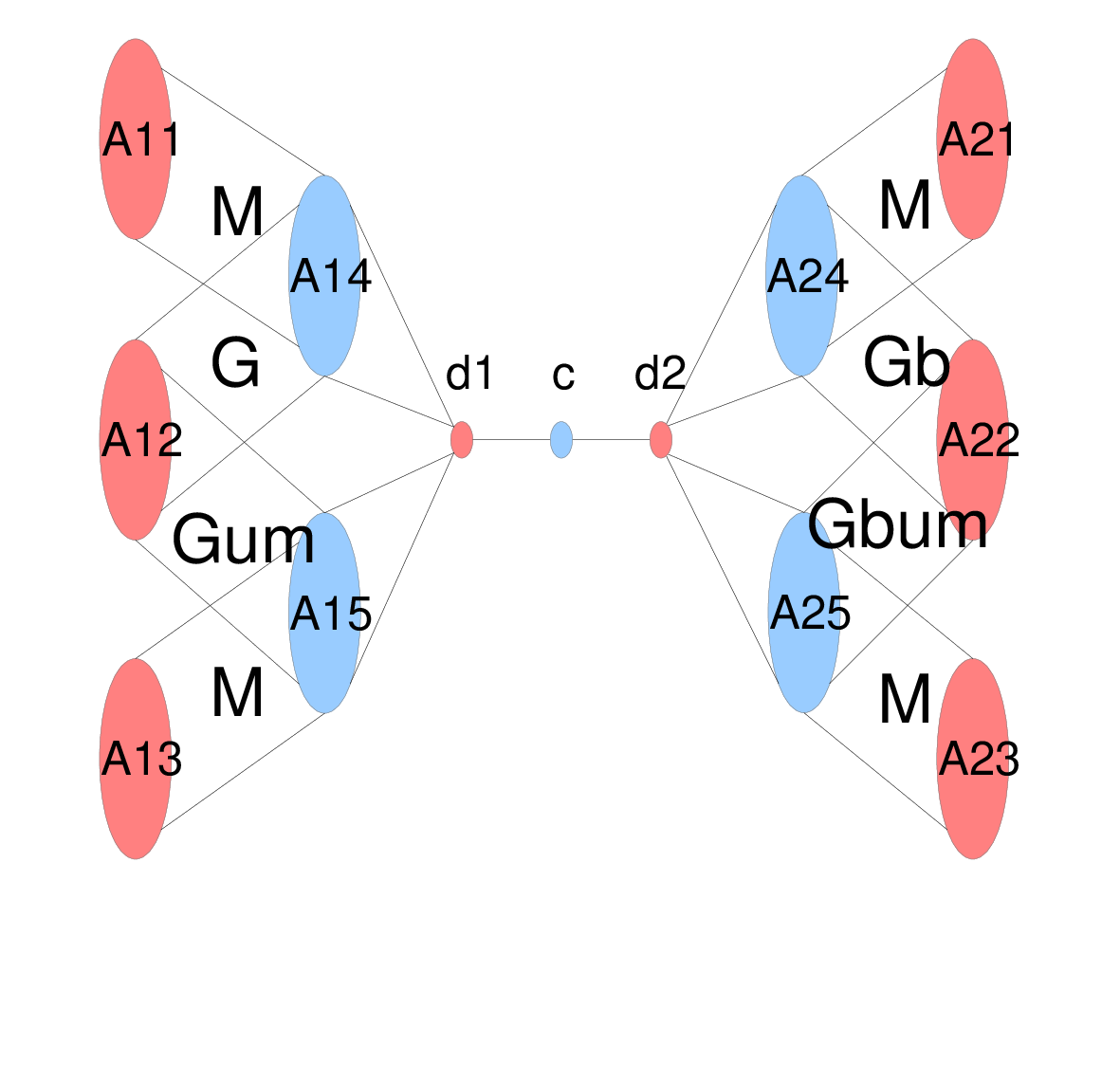}
	\end{psfrags}
	\vspace{-1.2cm}

\noindent If $(G,k)$ is a positive instance of \textsc{Quantum Threshold At Least} $n{-}k$, since $\kappa_Q(G){=}\max(\kappa(G),\kappa(\overline{G}))$,  either $\kappa(G)$ or $\kappa(\overline{G})$ is greater than $n{-}k$:\\

$\bullet$ If $\kappa(G)\geq n{-}k$, there exists $C$ such that $|Odd(C)| \geq n{-}k$. let $C'\subseteq R$ be $\{ a_{1,2,u}, u \in C \} \cup \{ a_{1,1,u},  u \in Even(C) \} \cup \{ a_{i,3,u},  u \in C{\oplus} Even(C) \} \cup \{d_2\}$ where ${\oplus}$ is the symmetric difference. Since $d_2 \in C'$, then $c$, $A_{2,4}$ and $A_{2,5}$ are in $Odd(C')$. Since $a_{1,2,u}a_{1,4,v}\in E_1\Leftrightarrow uv\in E$ and $\{ a_{1,2,u},  u \in C \}\subseteq C'$, then $\{ a_{1,4,u},  u \in Odd(C)\} \subseteq Odd(C')$. $A_{1,1}$ is connected to $A_{1,4}$ by a matching and $\{ a_{1,1,u},  u \in Even(C) \} \subseteq C'$ so $\{ a_{1,4,u},  u \in Even(C)\} \subseteq Odd(C')$ thus $A_{1,4} \subseteq Odd(C')$. Since $A_{1,5}$ is connected to $A_{1,2}$ like $A_{1,4}$ plus a matching, $\{ a_{1,5,u},  u \in Odd(C){\oplus} C\}\subseteq Odd(C'\cap A_{1,2})$ and $\{ a_{1,5,u},  u \in Even(c){\oplus} C\}\subseteq Odd(C'\cap A_{1,3})$, thus $A_{1,5} \subseteq Odd(C')$. So $B \subseteq Odd(C')$, and $|C'|=|C|+|Even(C)|+|Even(C){\oplus} C|+1=2|Even(C)\cup C|+1 \leq 2k+1=k'$, thus $(G',k')$ is a positive instance of \textsc{Oddset}. \\
$\bullet$  If $\kappa(\overline{G}){\geq} n{-}k$ then $d_1\in C'$ instead of $d_2$ and using the same process with neighbourhood relations of $\overline{G}$ instead of $G$ we obtain $(G',k')$ is a positive instance of \textsc{Oddset}.

\indent If $(G',k')$ is a positive instance of \textsc{Oddset}, there exists $C'\subseteq R$ such that $B\subseteq Odd(C')$. $c$ is dominated either by $d_1$ or $d_2$. If $d_2 \in C'$, then $A_{2,4}$ and $A_{2,5}$ are in $Odd(C')$. Since $A_{1,4}$ is connected to $A_{1,1}$ by a matching $A_{1,2}\cap C'\ne \emptyset$, let $C= \{u, a_{1,2,u} \in C' \}$ be a set of vertices in $V$, so similarly $|C'|=2|Even(C)\cup C|+1$ which implies that $|C'|=2|Odd(C)|+1$, so $\kappa(G)\ge n{-}k$ thus $(G,k)$ is a positive instance of \textsc{Quantum Threshold Of Size At Least} $n{-}k$. If $d_1 \in C'$, similarly $\kappa(\overline{G})\ge n{-}k$, so $(G,k)$ is a positive instance of \textsc{Quantum Threshold At Least} $n{-}k$. \hfill $\Box$
\end{proof}

\vspace{-0,4cm}

\begin{coro}
All the following problems: \textsc{WOD Set Of Size At Least} $n{-}k$, \textsc{Non-WOD Set Of Size At Most} $k$ and \textsc{Quantum Threshold At Least} $n{-}k$ even for bipartite graph are \textup{FPT}-equivalent to \textsc{Oddset} and thus are all hard for \textup W$[1]$ and in \textup W$[2]$.
\end{coro}

\vspace{-0.5cm}

\section{Approximability}\label{sec:approx}

Weak odd domination problems \textsc{WOD Set Of Size At Least} $n{-}k$, \textsc{non-WOD Set Of Size At Most} $k$, and \textsc{Quantum Threshold At Least} $n{-}k$ are not fixed parameter tractable unless \textup{FPT} = \textup W$[1]$ (see section \ref{sec:PC}).  In this section we consider the question of the existence of an approximation algorithm for optimisation problems associated with the weak odd domination:

\vspace{0.3cm}

\indent {\begin{tabular}{l}
 \textsc{Max WOD Set}\\
   input: A graph $G=(V,E)$\\
   cost function: $D\mapsto |Odd(D)|$\\
   output: $D{\subseteq} V$ s.t. $|Odd(D)| {=}\kappa(G)$~~~~\\
  \end{tabular}}
  
  \vspace{0.3cm}
  
  \indent {\begin{tabular}{l}
\textsc{Min non-WOD Set}\\
input: A graph $G=(V,E)$\\
cost function $D\mapsto |D\cup Odd(D)|$\\
 output: $D{\subseteq} V$ s.t. $|D\cup Odd(D)| {=}\kappa'(G)$\\
  \end{tabular}}

\vspace{0.3cm}
  

\indent {{\begin{tabular}{l}\textsc{Max non-Accessible Set}\\
   input: A graph $G=(V,E)$\\
   cost function $D\mapsto \max(|Odd_G(D)|,|Odd_{\overline G}(D)|)$\\
   {output: $D\subseteq V$ s.t. $\max(|Odd_G(D)|,|Odd_{\overline{G}}(D)|) =\kappa_Q(G)$} 
    \end{tabular}}}

~\\Notice that \textsc{Max non-Accessible Set} is in $\textup{APX}$ as it admits the $2$-approxima\-tion algorithm  which consists in outputting  any vertex of the graph. Indeed, given a graph $G{=}(V,E)$ of order $n$, $\forall v{\in} V$, $\kappa_Q(G){\ge} \max(|N(v)|, |V{\setminus} N(v)| {-}1)\ge\frac{n{-}1}2$. Since $\kappa_Q(G){\leq} n{-}1$, it produces a $2$-approximation. 

In the rest of the section, we prove the \textup{APX}-completeness of \textsc{Max WOD Set} and \textsc{Min non-WOD Set}, which implies the non-existence of polynomial approximation scheme for these problems unless \textup{P=NP}. 

\vspace{-0,3cm}

\subsection{Maximum WOD Set}

\begin{theo}\label{MaxWOD:Hard}
The \textsc{Max WOD Set} problem is complete for \textup{APX}.
\end{theo}

The proof of the harness consists in a reduction from \textsc{Max 3-Sat $B$} proved to be hard for $\textup{APX}$ in \cite{papa}, for any $B\in \mathbb N$,

\vspace{0.3cm}

\indent \textsc{Max 3-Sat $B$}\\
\indent input: a 3-CNF formula $\psi$ where every variables occurs at most $B$ times\\
\indent output: an assignment of variable of $\psi$ s.t. the maximum number of clauses\\\indent\qquad is satisfied\\

\begin{lemm}
For any  $B{\ge} 0$  \textsc{Max WOD Set} is harder than \textsc{Max 3-Sat $B$} by an $L$-reduction.
\end{lemm}

\begin{proof}
Given $\psi$ an instance of \textsc{Max 3-Sat $B$} with $n$ clauses, let $f$ such that $f(\psi)=G'$ (see Figure) be an instance of \textsc{Max WOD Set} such that $G'=(C\cup D\cup F,E_1\cup E_2)$ where:\\
$\begin{array}{l l}
C=\{a,\overline{a}|a$ is a variable of $\psi\} & E_1=\{ad_{i,a},ad_{i,ab},ad_{i,ac},\\
F=\{f_{a,j}|a$ is variable$ & \qquad\quad ad_{i,abc}|i\in [1,n]\}\\
\qquad\quad$ of $\psi,j\in[1,4B+1]\} & E_2=\{af_{a,j},\overline{a}f_{a,j}|\\
D_i=\{d_{i,a},d_{i,b},d_{i,c},d_{i,ab},d_{i,ac},d_{i,bc},d_{i,abc}|a,b,c & \qquad\quad j\in [1,4B+1],F_{i,j}\subseteq F\}\\
\qquad\quad$ are the literals of the $i^{th}$ clause of $\psi\}\\
D=\bigcup\limits_{i\in[1,n]} D_i\\
\end{array}$ 

\vspace{-0.2cm}
	
	\hspace{1cm}
	\begin{psfrags}
	\psfrag{ai}{$d_{i,a}$}
	\psfrag{bi}{$d_{i,b}$}
	\psfrag{ci}{$d_{i,c}$}
	\psfrag{abi}{$d_{i,ab}$}
	\psfrag{aci}{$d_{i,ac}$}
	\psfrag{bci}{$d_{i,bc}$}
	\psfrag{abci}{$d_{i,abc}$}
	\psfrag{a}{$a$}
	\psfrag{b}{$b$}
	\psfrag{c}{$c$}
	\psfrag{ab}{$\overline{a}$}
	\psfrag{bb}{$\overline{b}$}
	\psfrag{cb}{$\overline{c}$}
	\psfrag{F}{$\bf F$}
	\includegraphics[width=10cm,height=5cm]{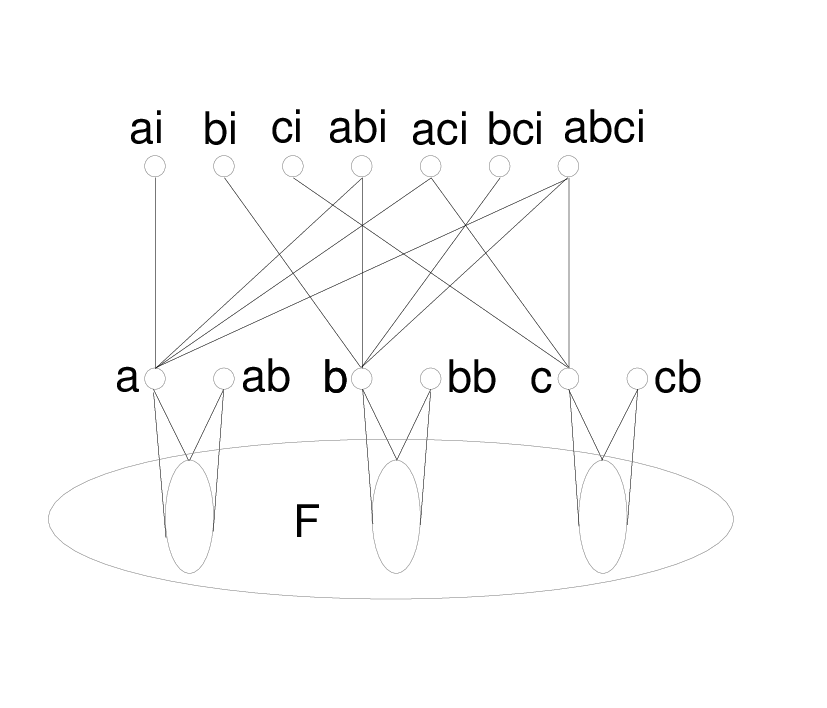}
	\end{psfrags}


Let $g$ such that given $X\subseteq V(G')$ a set of $G'$, $g(X)$ is an assignment of $\psi$ such that the variable $a$ is true if the vertex $a\in D$ (even if $\overline{a}\in D$) and is false otherwise. $g$ is computable in polynomial time and since $G'$ is polynomial in $|\psi|$, $f$ is computable in polynomial time.

In $G'$, $|C|\leq 6n$ and $|D|=7n$, and since each variable of $\psi$ is in a most $B$ clauses then $|F|\le 3n\times 4B+1$, so the size of $G'$ is linear in the size of $\psi$. Let $MaxSAT(\psi)$ be the value of the number of true clauses in the optimal solution of \textsc{Max 3-Sat $B$} on $\psi$, since $|G'|=\alpha|\psi|$ and that $MaxSAT(\psi)\ge \frac{|\psi|}{2}$, $\kappa(G')\le 2\alpha MaxSAT(\psi)$.

Given $C'\subseteq C$, let $(a\lor b\lor c)$ be the $i^{th}$ clause of $\psi$, $C'\cap \{a,b,c\}\ne\emptyset$ if and only if $|Odd(C')\cap D_i|=4$ since:\\
$\bullet$ If $|C'|=1$, let $a\in C'$ then $d_{i,a},d_{i,ab},d_{i,ac},d_{i,abc}\in Odd(C')$ and\\ $d_{i,b},d_{i,c},d_{i,bc}\in Even(C')$ and symmetrically for $b$ and $c$.
\\$\bullet$ If $|C'|=2$, let $a,b\in C'$ then $d_{i,a},d_{i,b},d_{i,ac},d_{i,bc}\in Odd(C')$ and\\ $d_{i,c},d_{i,ab},d_{i,abc}\in Even(C')$ and symmetrically for $b,c$ and $a,c$.
\\$\bullet$ If $|C'|=3$, let $a,b,c\in C'$ then $d_{i,a},d_{i,b},d_{i,c},d_{i,abc}\in Odd(C')$ and\\ $d_{i,ab},d_{i,ac},d_{i,bc}\in Even(C')$.

For all variable $a$ of $\psi$, there exists $F_a=\{f_{a,j}|j\in [1,4B+1]\}$ an independent set of size $4B+1$ connected only to $a$ and $\overline{a}$. Given $C'\subseteq C$, since $|Odd(C')\cap D_i|\leq 4$ and that a variable $a$ appears in at most $B$ clauses in $\psi$, if $a,\overline{a}\in C'$, $|Odd(C'\setminus\overline{a})|>|Odd(C')|$ thus there is no $a,\overline{a}$ pairs in the optimal solution of \textsc{Max WOD Set}.

Let $X_{opt}$ be the optimal solution  of \textsc{Max WOD Set}, since there is no $a,\overline{a}$ pairs in $X_{opt}$, then $F\subseteq Odd(X_{opt})$ and the number of $D_i$ dominated by $X_{opt}$ is $MaxSAT(\psi)$. Let $X$ be a solution of \textsc{Max WOD Set} and $l$ the number of satisfied clauses by $g(X)$. $MaxSAT(\psi)-l=k$ implies that $g(X)$ satisfies $k$ minus clauses than $g(X_{opt})$. Thus $X$ oddly dominates at least $k$ minus $D_i$ than $X_{opt}$ or contains $a,\overline{a}$, so $X$ oddly dominates at least $4k$ minus vertices than $X_{opt}$. This implies that $MaxSAT(\psi)-l\leq\beta(\kappa(G')-|Odd(X)|)$. \hfill $\Box$
\end{proof}

\noindent \textit{\bf Proof of Theorem \ref{MaxWOD:Hard}.} For all $B\geq 3$ \textsc{Max 3-Sat $B$} is hard for $\textup{APX}$ and there exists an $L$-reduction from \textsc{Max 3-Sat $B$} to \textsc{Max WOD Set}, thus \textsc{Max WOD Set} is also hard for $\textup{APX}$. The completeness is obtained by applying the conditional probabilities method \cite{Raghavan1988130} on the linear bound of $\kappa(G)$ given by Lemma \ref{lem:bound} which gives a polynomial 4-approximation of $\kappa$.


\begin{coro}
There exists $\epsilon>0$ such that there is no polynomial time $(1+\epsilon)$-approximation of \textsc{Max WOD Set} unless \textup {\textup{P=NP}}.
\end{coro}


\subsection{Minimum non-WOD Set}

\begin{theo}\label{MinNWOD:hard}
The \textsc{Min non-WOD Set} problem is complete for \textup{APX}.
\end{theo}

The proof of the hardness consists in a reduction from \textsc{Max WOD Set} proved to be hard for $\textup{APX}$ in Theorem \ref{MaxWOD:Hard}.

\begin{lemm}
\textsc{Min non-WOD Set} is harder than~\textsc{Max WOD Set}~by~an~$L$-reduction.
\end{lemm}

\begin{proof}
Given $G=(V,E)$ an instance of {Max WOD Set} of size $n$, let $f$ such that $f(G)=G'$ be an instance of \textsc{Min non-WOD Set} such that $G'=(V\cup d,E_1\cup E_2)$ where $E_1=\{uv| u\in V,v\in V,uv\notin E\}$ and  $E_2=\{ud| u\in V\}$. Let $g$ such that given $X\subseteq V(G')$ a non-WOD set of $G'$, $g(X)$ is a WOD set of $V(G)$ defined by $g(X)=\{u| u\notin X, u\ne d\}$. Since $X$ is a non-WOD set there exists $D\subseteq X$ such that $|D|= 1\bmod 2$ and $Odd(D)\subseteq X$, so $\overline{X}\subseteq Even(D)$, thus in $\overline{G'}$ $\overline{X}\subseteq Odd(D)$. Since $d$ is isolated in $\overline{G'}$, $\overline{X}\subseteq Odd(D)$ in $G$ thus $g(X)$ is a WOD set in $G$ and $|g(X)|=n-|X|$. $g$ is computable in polynomial time and since $G'$ is polynomial in $|\psi|$, $f$ is computable in polynomial time. 
By Lemma \ref{lem:stgdual}, $\kappa'(G')+\kappa(\overline{G'})= n+1$, since $\overline{G'}$ is $G$ and an isolated vertex and that an isolated vertex does not change the value of $\kappa$, then $\kappa'(G')=n+1-\kappa(G)$. By Lemma \ref{lem:bound}, $\kappa(G)\geq \frac{n}{4}$, thus $\kappa(G')\leq \alpha\kappa(G)$. 
By adding $|g(X)|=n-|X|$ to $\kappa'(G')=n+1-\kappa(G)$ we have that for all solution $X$ of \textsc{Min non-WOD Set}, $|\kappa(G)-|g(X)||\leq \alpha|\kappa'(G')-|X||$. \hfill $\Box$
\end{proof}

\noindent \textit{\bf Proof of Theorem \ref{MinNWOD:hard}.} \textsc{Max WOD Set} is hard for $\textup{APX}$ and there exists an $L$-reduction from \textsc{Max WOD Set} to \textsc{Min non-WOD Set}, thus \textsc{Min non-WOD Set} is also hard for $\textup{APX}$. Similarly to \textsc{Max WOD Set} the completeness is given by applying the conditional probabilities method \cite{Raghavan1988130} on Corollary \ref{cor:bound} which gives a polynomial 8-approximation of $\kappa'$.

\begin{coro}
The exits $\epsilon>0$ such that there is no polynomial time $(1+\epsilon)$-approximation of \textsc{Min non-WOD Set} unless \textup{\textup{P=NP}}.
\end{coro}

\section{Conclusion}

In this paper we have explored the parameterized complexity and approximation of weak odd domination problems: largest WOD set, smallest non-WOD set, quantum threshold, and their respective related quantities $\kappa$, $\kappa'$ and $\kappa_Q$. 
We have proved bounds on these quantities. These bounds imply that for any of these problems, there is a parameterization which is \textup{FPT}. However, if one considers a more natural parameterization of these problems, it turns out that all the variants of weak odd domination problems are hard for \textup W$[1]$ and in \textup W$[2]$ even for bipartite graphs. More precisely we show that they are equivalent to \textsc{Oddset}. Regarding the approximation we have proved that (1) maximum WOD set and minimum non-WOD set are complete for $\textup{APX}$, so they have no polynomial approximation scheme unless $\textup{P=NP}$; and (2) maximum non-accessible set is in \textup{APX} with a trivial constant factor approximation. 
The existence of a polynomial approximation scheme for maximum non-accessible set, which would be of great interest for the design of graph-based quantum secret sharing schemes, remains open.
Regarding the parameterized complexity, \emph{tight} bounds on $\kappa$, $\kappa'$ and $\kappa_Q$ would allow for above guarantee technics  to tackle the existence of \textup{FPT}-algo\-rithms only based on the existence of bounds on the parameter (see \cite{mahajan} details).

\vspace{-0.4cm}

\subsubsection*{Acknowledgements.}

The authors want to thank Sylvain Gravier, J\'er\^ome Javelle, and Mehdi Mhalla for fruitful discussions and the anonymous reviewer for his useful comments.  This work has been funded by the ANR-10-JCJC-0208 CausaQ grant and by R\'egion Rh\^one-Alpes.

\end{document}